\documentclass[11pt]{article}
\usepackage[T1]{fontenc}
\usepackage{lmodern,microtype}
\usepackage[margin=2.5cm]{geometry}
\usepackage{amsmath,amssymb,amsthm}
\usepackage{graphicx,flafter}
\usepackage{placeins}
\usepackage[authoryear,round]{natbib}
\usepackage[hidelinks]{hyperref}
\newcommand{\E}{\mathbb E}
\newcommand{\Law}{\mathcal L}
\newcommand{\Unif}{\operatorname{Unif}}
\newcommand{\cx}{\leq_{\mathrm{cx}}}
\newcommand{\1}{\mathbf 1}
\newtheorem{theorem}{Theorem}
\newtheorem{corollary}{Corollary}
\newtheorem{claim}{Claim}
\theoremstyle{definition}

\hypersetup{
pdftitle={Two-Sided Plausibility Constraints in Persuasion},
pdfauthor={Constantine Sorokin}}
\title{\vspace{-1em}Two-Sided Plausibility Constraints in Persuasion}
\author{Constantine Sorokin\thanks{I thank Dirk Bergemann, Tibor Heumann, and Stephen Morris
for suggesting that two-sided posterior-mean constraints might fail to capture
the receiver's initial information. I also thank Eyal Winter, Andriy Zapechelnyuk,
and Fedor Sandomirskiy for fruitful discussions and useful comments.
\newline The idea behind the discrete example in Section~\ref{sec:discrete-example} is due to
Tibor Heumann, whom I thank for sharing it in private communication.}\\[3pt]\small University of Glasgow\\[2pt]\small\href{mailto:Constantine.Sorokin@glasgow.ac.uk}{Constantine.Sorokin@glasgow.ac.uk}}
\date{October 9, 2026}

\begin{document}
\maketitle

\begin{abstract}
A sender with commitment provides additional information to an already informed receiver. Final posterior means must be more dispersed than initial means and less dispersed than the state, in convex order. We characterize when these two plausibility constraints are also sufficient, so that every distribution satisfying them can be induced. If the state has a positive density, they are sufficient exactly under interval-partition information: the receiver learns which interval of a partition contains the state and nothing more, where an interval may be a single point. Otherwise, a nonempty open set of distributions satisfying both constraints cannot be induced, even when the sender observes the receiver's information. We use the sufficient case to derive an optimal symmetric disclosure rule in a second-price auction where bidders initially know valuation intervals.
\end{abstract}

{\small\noindent\textit{JEL:} D82, D83, D44.\par
\par}

\section{Introduction}\label{sec:introduction}

A receiver's initial information restricts what further persuasion can
achieve. We study a sender who commits to a message rule about an uncertain state.
The sender may observe the receiver's initial signal and condition the
message on it.
The receiver knows the rule and combines each message with her initial
information. We ask which distributions of her \emph{posterior mean}---the
expected state given her information---the sender can induce. When the
receiver's payoff from each action is affine in the state, her posterior
mean suffices to evaluate those actions. In our auction application,
it determines her bid.

For a receiver who initially knows only the prior distribution of the
state, a single upper bound describes all feasible distributions of
posterior means. A feasible distribution must be less dispersed than the state
distribution in \emph{convex order}: it has the same mean and gives no
larger expectation to any convex function. By \citet{Strassen}, every
distribution satisfying this bound can be generated by some message
rule. The sender can therefore choose a distribution of posterior means
without first constructing the rule that implements it
\citep{GentzkowKamenica2016}.

An informed receiver retains what she already knows, so final posterior
means must also be at least as dispersed as initial posterior means.
It is tempting to add this lower bound and optimize over distributions
between the two bounds. We show that the bounds describe the feasible set
if and only if the receiver initially learns which interval of a
partition contains the state, and nothing more; an interval may be a
single point, in which case she knows the state. Outside this class,
the bounds admit unattainable distributions, and the failure persists
under small changes to the target.

The missing restriction is the receiver's conditional distribution of
the state after each initial observation. If an observation rules out
states above a given threshold, further disclosure cannot produce a
posterior mean above that threshold
(Section~\ref{sec:discrete-example} gives an example).
The distribution of initial means need not record this restriction: it
says how often each mean occurs, without specifying which states a
receiver with that mean considers possible. Each bound alone can be
represented by a joint distribution: one links initial and final means,
the other final means and states. In each case, the more dispersed
variable has conditional mean equal to the less dispersed one.
Combining these joint distributions need not reproduce the receiver's
actual joint distribution of initial means and states. Such a
construction would describe a different initial experiment.

Feasible disclosure can always be described separately after each initial
observation: the receiver's belief at that observation plays the role of
the prior, and the single upper bound applies. For a fixed payoff function of the posterior mean, the sender
can also optimize these conditional problems separately. Averaging their
outcomes gives the unconditional distribution of final means. Describing
that distribution directly, without solving the conditional problems,
is especially useful when disclosure affects
other agents who do not observe the receiver's information. In a
second-price auction, for example, further disclosure to one bidder
changes the distribution of the bid her competitors face. Revenue depends on how bids drawn from
the unconditional distributions compare.

To state the results, let $\theta$ be the state and $X$ the receiver's initial
posterior mean. The initial experiment is specified by a fixed joint law
$\pi$ of $(X,\theta)$ on $[a,b]^2$, where $a<b$ are finite and
$\E_\pi[\theta\mid X]=X$. The receiver observes $X$, and the sender may
randomize a message $M$ conditional on the realized pair $(X,\theta)$. The
receiver uses both observations, so her final posterior mean is
$Y=\E[\theta\mid X,M]$. Write
$F=\Law(\theta)$ for the prior, $G_0=\Law(X)$ for the distribution of
initial posterior means, and $G=\Law(Y)$ for the distribution of final
posterior means, where $\Law$ denotes probability law. A target
distribution $G$ is \emph{implementable} if some such message rule
produces it; let $\mathcal I(\pi)$ be the set of implementable
distributions. Denoting convex order by $\cx$, the necessary bounds are
\begin{equation}\label{eq:band}
\underbrace{G_0\cx G}_{\text{final information includes }X}
\qquad\text{and}\qquad
\underbrace{G\cx F}_{\text{feasible from the prior alone}}.
\end{equation}
We call these bounds the \emph{two-sided plausibility constraints}.
They are posterior-mean implications of Bayes plausibility: final beliefs
must average back to initial beliefs, which in turn average to the prior.
Together they define the \emph{two-sided band}
$\mathcal B(G_0,F)=\{G:G_0\cx G\cx F\}$, which contains
$\mathcal I(\pi)$. The two bounds are sufficient when the converse also
holds: every distribution in the band is implementable, so that
$\mathcal I(\pi)=\mathcal B(G_0,F)$.

One form of initial information makes the two bounds sufficient. Choose
at most countably many disjoint open intervals $I_j\subseteq(a,b)$, each
with positive probability. The initial law $\pi$ generates \emph{interval-partition
information} if, $\pi$-almost surely,
\begin{equation}\label{eq:partition}
X=\E_\pi[\theta\mid \theta\in I_j]\quad\text{if }\theta\in I_j,
\qquad
X=\theta\quad\text{if }\theta\notin\textstyle\bigcup_j I_j .
\end{equation}
The receiver learns which interval contains the state and nothing more,
and learns the state itself outside all the intervals. The class includes
both extremes: with no intervals she knows the state, and with the single
interval $(a,b)$ she knows only the prior. Within each disclosed interval,
her belief is the prior conditioned on that interval.
The upper bound alone describes feasible distributions
there, as for an initially uninformed receiver. The two bounds ensure
that a target preserves each interval's probability and mean, so these
conditional choices can be combined to produce every target in the band.
The first theorem shows that no other initial experiment has this property.

\begin{theorem}[When the two bounds are sufficient]\label{thm:main}
Suppose that $F$ has a density $f$ with $f>0$ Lebesgue-almost
everywhere on $(a,b)$. Then $\mathcal I(\pi)=\mathcal B(G_0,F)$ if and only if
the fixed initial joint law $\pi$ generates interval-partition information
as in \eqref{eq:partition}. For every other initial experiment, the band
contains an absolutely continuous law that is not implementable.
\end{theorem}

The main step is the converse: if every target in the band is
implementable, the initial information must be an interval partition.
Consider a target in the band that pools a short interval of states at
the prior's conditional mean on that interval and reveals all other states.
Any implementation must pool those same states. A receiver who gets the
pooling message evaluates the pooled states with her own belief, but the
target requires her final posterior mean to equal the prior's mean of
the interval. Every receiver who can get that message must therefore
agree with the prior about the interval's mean. If this agreement holds
for every such pooling target, the initial information must be an
interval partition.
The impossibility extends to richer initial signals, even when the sender
observes them (Section~\ref{sec:richer}).

Call a target \emph{excluded} if it belongs to the band but cannot be
implemented. Exclusion is robust: a sufficiently small change within
the band leaves the target excluded. In particular, mixing an excluded
target with any other target in the band keeps it excluded when the
weight on the other target is small enough, and one bound on that weight
works for every other target.

\begin{theorem}[Robustness of exclusion]
\label{thm:excluded-set}
Under the assumptions of Theorem~\ref{thm:main}, suppose $\pi$ does not
generate interval-partition information. For every
$G_*\in\mathcal B(G_0,F)\setminus\mathcal I(\pi)$, there is
$\bar s\in(0,1]$ such that
\[
(1-s)G_*+sG
\in\mathcal B(G_0,F)\setminus\mathcal I(\pi)
\qquad\text{for all }G\in\mathcal B(G_0,F),\quad 0<s<\bar s.
\]
For each such $s$, these mixtures form a copy of the band shrunk toward
$G_*$ by the factor $s$, and every target in this copy is excluded. Moreover,
$\mathcal B(G_0,F)\setminus\mathcal I(\pi)$ is a nonempty, relatively
open subset of $\mathcal B(G_0,F)$ in the weak topology, and it has
positive probability under every Borel probability distribution on
$\mathcal B(G_0,F)$ whose support is $\mathcal B(G_0,F)$.
\end{theorem}

The excluded set thus contains a neighborhood of each of its targets
within the band. If targets are drawn from any distribution that gives
positive probability to every nonempty open subset of the band, the
chance of drawing an excluded target is positive.
The discrete example in Section~\ref{sec:discrete-example} gives an
explicit bound on the mixing weight.

For optimization, the band therefore gives an exact constraint under
interval-partition information and a relaxation otherwise. A maximizer
of the relaxed problem solves the sender's problem whenever it is
implementable; outside interval partitions, this must be checked against
the initial experiment (Section~\ref{sec:optimization}). An auction
application illustrates the sufficient case: when bidders initially know
valuation intervals, an optimal symmetric rule reveals low valuations
and pools high valuations within each interval.

Section~\ref{sec:literature} relates these results to the literature.
Section~\ref{sec:implementation} completes the model, describes
implementation separately for each initial posterior mean, and extends
the failure to receivers with richer initial information.
Section~\ref{sec:examples} develops the two examples, and
Section~\ref{sec:auctions} applies the sufficient case to auction disclosure.
Section~\ref{sec:discussion} discusses the broader implications, and
Section~\ref{sec:proof} contains the proofs.

\section{Related literature and contribution}\label{sec:literature}

\citet{Blackwell} ranks experiments by garbling, which reduces the
dispersion of posterior means in convex order. \citet{Strassen} shows
that $\tau\cx\tau'$ holds precisely when there is a \emph{martingale
coupling}: a joint distribution with marginals $\tau,\tau'$ in which
the conditional expectation of the second variable given the first
equals the first. Applied to~\eqref{eq:band}, this gives separate
couplings of $(X,Y)$ and $(Y,\theta)$. Our question is when they can
be combined while preserving the fixed law $\pi$ of $(X,\theta)$.

The closest result is the characterization of unique martingale
couplings in \citet[Theorem~2.5 and Remark~2.6]{NutzWangZhang}.
Under our assumptions, their result implies that $G_0$ and $F$ uniquely
determine a joint law satisfying $\E[\theta\mid X]=X$ exactly under
interval-partition information. This gives sufficiency in
Theorem~\ref{thm:main}: combining the two couplings from Strassen's
theorem must reproduce that unique law. Our contribution is necessity.
If every target in the band can be implemented while preserving $\pi$,
then $\pi$ must have the interval-partition form. Targets with densities
already suffice to establish this conclusion, even under the weaker
constraints that apply to richer initial signals
(Section~\ref{sec:richer}).

\citet[Propositions~2.3--2.4]{Sester} asks when prescribing an
intermediate marginal leaves a martingale transport optimum unchanged.
There the joint law of the first and last variables can be chosen among
optimal couplings. Here it is fixed by the receiver's initial experiment.

In persuasion, \citet{KamenicaGentzkow} formulate the sender's choice
over distributions of posterior beliefs, meaning conditional
distributions of the state, subject to Bayes plausibility: the beliefs
must average to the prior. When payoffs depend on posterior means,
\citet{DworczakMartini} and \citet{KleinerMoldovanuStrack} develop
methods for optimizing over distributions satisfying convex-order
constraints. With an informed receiver, these methods apply separately
at each initial posterior mean, using the receiver's conditional belief
as the prior. This does not require interval-partition information.
Our results identify when the two bounds alone describe the distributions
obtained by combining these conditional choices. This matters, for
example, when auction revenue depends on the resulting bid distribution
(Section~\ref{sec:auctions}).

A related literature studies persuasion when the sender cannot observe
the receiver's private information. \citet{KolotilinMylovanovZapechelnyukLi}
compare mechanisms that condition disclosure on the receiver's reported
preference type with experiments that do not use such reports.
\citet{Kolotilin2018} develops a linear
programming approach to optimal disclosure with a privately informed
receiver, and \citet{GuoShmaya} characterize disclosure that induces
acceptance of a project when the receiver has private information about the state.
In our problem, the sender may observe the receiver's information and
condition the message on it, so the failure of the two bounds does not
stem from uncertainty about what the receiver knows.

Interval partitions also appear in \citet{Censorship}
and \citet{MonotonePersuasion} as experiments chosen by the sender; in the
terminology of \citet{MonotonePersuasion}, interval-partition information
is a monotone signal. Here such a partition describes what the receiver
already knows, and the sender chooses what further information to
provide.

\citet[Section~4.4]{AugiasUhe} solve persuasion problems with exogenous
lower and upper bounds on informativeness. These bounds can represent a
minimum disclosure requirement and a limit on what can be revealed.
Their optimization ranges over distributions between two convex-order
bounds. The present results
identify when their problem with the lower bound $G_0$ and the upper
bound $F$ is also the problem of a sender facing a receiver with the
initial experiment $\pi$: exactly under interval-partition information.
Otherwise, their solution bounds the sender's value from above and is
optimal whenever it is implementable (Section~\ref{sec:optimization}).

When the receiver's full belief, rather than her posterior mean,
describes her information, the difficulty studied here does not arise.
\citet[p.~647]{LiNorman} use the fact that any mean-preserving spread of
the current distribution of beliefs can be induced by further disclosure,
so that each sender's problem in sequential persuasion separates belief
by belief. \citet{LeskelaVihola} prove a conditional version of
Strassen's theorem with measurable couplings; it supplies the description
of implementation separately for each initial posterior mean in
Section~\ref{sec:implementation}. Theorem~\ref{thm:main} identifies when
that description can be replaced by the two bounds on distributions of
posterior means.

\section{Model and information constraints}\label{sec:implementation}

We now derive the exact restrictions on implementation, conditional on
each initial posterior mean. Theorem~\ref{thm:main} identifies when the
two bounds capture these restrictions. Retain the fixed initial law $\pi$ and notation from the
introduction. All probability laws are Borel laws, and $\E_\pi$ denotes
expectation under $\pi$. The sender's message rule can randomize freely
conditional on $(X,\theta)\in[a,b]^2$. Formally, the rule is a measurable
probability kernel from $[a,b]^2$ to a standard Borel message space. For probability laws
$\tau,\tau'$ on $[a,b]$, write $\tau\cx\tau'$ when
$\int\varphi\,d\tau\leq\int\varphi\,d\tau'$ for every continuous convex
$\varphi:[a,b]\to\mathbb R$; affine functions of both signs give equality
of means. A law is absolutely continuous if it has a density with respect
to Lebesgue measure, and atomless if it assigns zero probability to every
point. Sets of laws carry the weak topology, in which
$G_n\to G$ when $\int h\,dG_n\to\int h\,dG$ for
every continuous $h:[a,b]\to\mathbb R$. Throughout the main results,
$F$ has a density $f$ with $f>0$ almost everywhere on $(a,b)$; the density
need not be continuous, and $G_0$ may have atoms.

Recall that the sender may condition the message on $X$ as well as on
$\theta$. This extra conditioning does not affect the characterization. Under interval-partition
information, the construction in Section~\ref{sec:sufficiency} uses only
the state. Otherwise, restricting what the sender observes can only
shrink $\mathcal I(\pi)$, so excluded targets remain excluded.

\subsection{Implementation for each initial posterior mean}

A target $G$ is implementable if and only if there is a triple
$(X,Y,\theta)$ in which $(X,\theta)$ has law $\pi$, $Y$ has law $G$, and
\begin{equation}\label{eq:implementation}
\E[Y\mid X]=X,\qquad \E[\theta\mid X,Y]=Y.
\end{equation}

The first identity says that, at each initial posterior mean, final
means average back to it. The second requires $Y$ to be the expected
state given both $X$ and $Y$. If $Y$ is the final posterior mean of an
implementation, both identities follow from iterated expectations.
Conversely, drawing
a message from the conditional law of $Y$ given $(X,\theta)$ and sending
$Y$ itself implements the target: the second identity makes the
posterior mean equal to the message.
Conditional Jensen gives both inequalities in~\eqref{eq:band}.

These requirements can be imposed separately for each initial posterior
mean. Let $F(\cdot\mid x)=\Law(\theta\mid X=x)$ be the receiver's belief
about the state after observing $x$. The distribution of final posterior
means at that observation must be below $F(\cdot\mid x)$ in convex order.
Consequently, a target $G$ is implementable if and only if there
is a measurable probability kernel
$x\mapsto G(\cdot\mid x)$ such that
\begin{equation}\label{eq:conditional-feasibility}
G=\int G(\cdot\mid x)\,G_0(dx),\qquad
G(\cdot\mid x)\cx F(\cdot\mid x)\quad\text{for }G_0\text{-almost every }x.
\end{equation}
Here $G(\cdot\mid x)$ is the distribution of final posterior means conditional
on $X=x$, and integrating over $G_0$ combines these conditional
distributions into the target $G$. Necessity follows by
conditioning~\eqref{eq:implementation} on $X$. For sufficiency, the
conditional martingale-coupling theorem of
\citet[Theorem~1.3]{LeskelaVihola} supplies measurable martingale
couplings of $G(\cdot\mid x)$ and $F(\cdot\mid x)$; integrating them against $G_0$
gives a triple in which $(X,\theta)$ has law $\pi$
and~\eqref{eq:implementation} holds. The question in
Theorem~\ref{thm:main} is when these conditional restrictions can be
replaced by the two inequalities~\eqref{eq:band}, which involve only the
distributions $G_0$, $G$, and $F$.

\subsection{Richer initial information}\label{sec:richer}

The receiver's initial information may be richer than her posterior
mean. Suppose that she observes a signal $S$ with values in a standard
Borel space, that $X=\E[\theta\mid S]$, and that $(X,\theta)$ has law $\pi$ as
before. The sender's message may depend on $(S,\theta)$, and the final
posterior mean is $Y=\E[\theta\mid S,M]$. Because $X$ is a function of $S$
and $Y$ is a function of $(S,M)$, iterated expectations give
\begin{equation}\label{eq:weak}
\E[Y\mid X]=X,\qquad \E[\theta\mid Y]=Y.
\end{equation}

The second identity is weaker than its counterpart in~\eqref{eq:implementation}: it
averages over different initial means that lead to the same final mean.
These two identities remain necessary whatever else the receiver knows.
Outside interval partitions, the proof of Theorem~\ref{thm:main} finds
an absolutely continuous target in the band for which no triple with
initial law $\pi$ satisfies even~\eqref{eq:weak}. That target is therefore
excluded under richer information too. The positive result continues
to apply when the conditional law of $\theta$ given $S$ depends on $S$
only through $X$, so that the extra information does not change beliefs
about the state.

\subsection{Atomless initial posterior means}

An atomless distribution of initial means, as arises from observing the
state plus independent normal noise, rules out interval-partition
information unless the state is already known. Each nondegenerate
disclosed interval puts mass $F(I_j)>0$ on its conditional mean.
Under interval-partition information, an atomless $G_0$ therefore permits
no such interval and requires $X=\theta$.

\begin{corollary}[Atomless initial means]\label{thm:impossibility}
Under the assumptions of Theorem~\ref{thm:main}, if $G_0$ is atomless
and $G_0\ne F$, the band contains an absolutely continuous law that is
not implementable. This law cannot be produced even when the receiver's
initial information is a richer signal $S$ as in
Section~\ref{sec:richer} and the sender observes $S$.
\end{corollary}

\subsection{Implication for optimization}\label{sec:optimization}

Optimizing over the band gives an upper bound on the sender's value for
any objective that evaluates the distribution of final posterior means.
Let $W(G)$ be such an objective,
continuous in the weak topology. Examples include the expected payoff
$W(G)=\int U\,dG$ for a continuous function $U$, and the auction revenue
in Section~\ref{sec:auctions}.\footnote{For a fixed number of bidders, the
joint law of independent bids varies weakly continuously with their common
law $G$. Revenue is weakly continuous because the second-highest bid is a
bounded continuous function of the bids.}
Every implementable target lies in the band,
so
\[
\max_{G\in\mathcal I(\pi)}W(G)
\leq\max_{G\in\mathcal B(G_0,F)}W(G).
\]
Both maxima exist: the band is weakly compact, as is the implementable
set (Claim~\ref{clm:closedness}). If a maximizer $G^{\mathrm{opt}}$ over the
band is implementable, it is also optimal over the implementable set.
By~\eqref{eq:conditional-feasibility}, implementation requires conditional
laws $G(\cdot\mid x)$ that lie below $F(\cdot\mid x)$ in convex order
and average to $G^{\mathrm{opt}}$. Under interval-partition information,
every target in the band meets these requirements, so optimization over
the band solves the original problem.

For an expected payoff $\int U\,dG$, optimization can also proceed
separately at each initial mean, using its conditional state distribution.
An objective such as auction revenue need not separate this way. Yet
even for an ordinary expected payoff, replacing the conditional
constraints by the band can raise the optimal value. The next corollary
shows that some continuous payoff has a strict gap whenever the initial
experiment is not an interval partition. The gap need not arise for
every payoff: affine payoffs have the same expectation throughout the band.

\begin{corollary}[When the band is exact for every continuous payoff]\label{cor:value}
Under the assumptions of Theorem~\ref{thm:main},
\[
\max_{G\in\mathcal I(\pi)}\int U\,dG
=\max_{G\in\mathcal B(G_0,F)}\int U\,dG
\quad\text{for every continuous }U:[a,b]\to\mathbb R
\]
if and only if $\pi$ generates interval-partition information.
\end{corollary}

\section{Examples}\label{sec:examples}

The examples show why the two bounds can hold even though the receiver's
initial information rules out the target. In both, the state is uniform
on $[0,1]$. The first example has two possible initial posterior means;
the second has a continuous distribution of initial posterior means
and positive conditional state densities.

\subsection{Discrete initial means: disclosed and noisy categories}
\label{sec:discrete-example}

If the receiver knows only whether $\theta<1/2$, the initial means are $1/4$
and $3/4$, each with probability $1/2$. This is interval-partition
information. Within either half, the sender faces an uninformed receiver
with the uniform conditional prior on that half. Every target between
the two bounds can be obtained by choosing a feasible distribution
separately in each half and combining the two.

Now suppose the receiver instead observes an initial signal
$S\in\{L,H\}$ that gives imperfect information about the half containing
the state:
\[
\Pr(L\mid \theta)=
\begin{cases}
1,&0\leq \theta<2/5,\\
1/2,&2/5\leq \theta\leq3/5,\\
0,&3/5<\theta\leq1.
\end{cases}
\]
Both signal realizations have probability $1/2$, and their posterior
means are $x_L=13/50$ and $x_H=37/50$. Observing $X$ therefore tells
the receiver which realization occurred. The conditional state supports
$[0,3/5]$ and $[2/5,1]$ overlap, so this initial information does not
come from an interval partition. With $\delta_c$ denoting unit
probability at $c$, consider the target $G$ in
\begin{equation}\label{eq:discrete-target}
G_0=\tfrac12\delta_{13/50}+\tfrac12\delta_{37/50},\qquad
G=\tfrac12\delta_{1/4}+\tfrac12\delta_{3/4}.
\end{equation}
The upper bound holds because, from the prior alone, disclosing which
half contains the state induces $G$. For the lower bound, replace $13/50$ by $1/4$ with probability
$49/50$ and by $3/4$ with probability $1/50$, reversing the probabilities
at $37/50$. This preserves each initial mean and produces $G$.
Yet a receiver who observed $L$ knows that the state does
not exceed $3/5$. No additional message can therefore produce a final
posterior mean of $3/4$ for that receiver. Since the target allows only
$1/4$ and $3/4$, it would require every receiver who observed $L$ to
end with mean $1/4$. Their final posterior means would then average to
$1/4$, contradicting the initial mean $13/50$.

Section~\ref{sec:verify-discrete} shows
that the target violates even the weaker conditions~\eqref{eq:weak},
which remain necessary when the receiver has richer initial information.
It also makes Theorem~\ref{thm:excluded-set} explicit: for every
$\eta\in\mathcal B(G_0,F)$ and every weight $0<s<3/175$, the mixture
$(1-s)G+s\eta$ is excluded.

\subsection{Continuous initial means and strictly ordered beliefs}
\label{sec:continuous-example}

The next example removes both the discreteness of the initial posterior
means and the restriction on the states each receiver considers possible.
Let the initial joint law have density
\begin{equation}\label{eq:joint}
f_{X,\theta}(x,\theta)=5\bigl[1+12(x-1/2)(\theta-1/2)\bigr],
\qquad x\in[2/5,3/5],\quad \theta\in[0,1],
\end{equation}
and zero elsewhere. It is strictly positive throughout the rectangle.
The marginals and conditional mean are
\[
X\sim\Unif[2/5,3/5],\qquad \theta\sim\Unif[0,1],\qquad \E[\theta\mid X]=X.
\]

Here $\Unif[u,v]$ denotes the uniform distribution on $[u,v]$.
Every receiver regards every interior state as possible. Beliefs are
strictly ordered by the monotone likelihood ratio property: the ratio
of the conditional density at a higher initial mean to that at a lower
mean increases strictly with the state.

Take the target
\begin{equation}\label{eq:target-cont}
G=\tfrac14\Unif[0,1/4]
+\tfrac12\Unif[199/400,201/400]
+\tfrac14\Unif[3/4,1].
\end{equation}
It has a piecewise constant density and satisfies both convex-order
bounds. The outer quarters of the state distribution are unchanged, while
the middle half is concentrated close to $1/2$.

\begin{figure}[htbp]
\centering
\includegraphics[width=\textwidth]{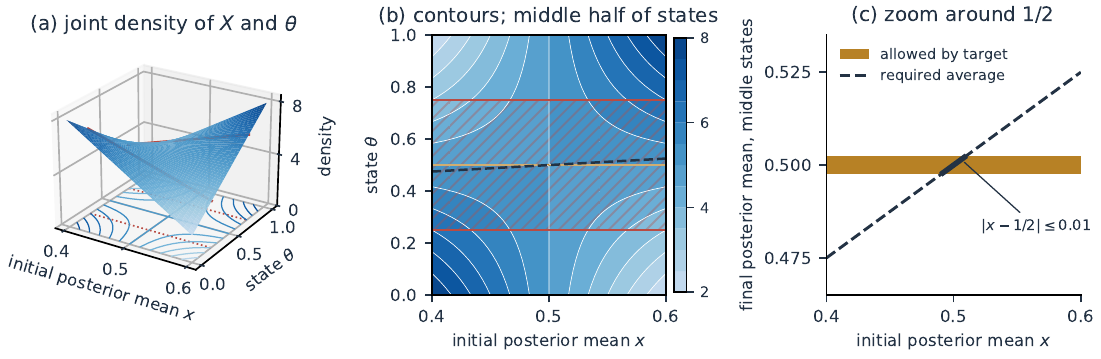}
\caption{The continuous example. (a)~The initial joint
density~\eqref{eq:joint} of the initial posterior mean $x$ and the state
$\theta$. (b)~Its contours. The hatched band is the middle half of the states,
$[1/4,3/4]$; the thin horizontal strip is the range $[199/400,201/400]$
that the target allows for their final posterior means; the dashed line
is $\frac12+\frac14(x-\frac12)$, the average state in the middle half for
a receiver with initial posterior mean $x$. (c)~Near $1/2$: the dashed
line is the required average of that receiver's final posterior means
for middle states. It meets the allowed strip only where
$|x-1/2|\leq1/100$ (thick segment).}
\label{fig:continuous}
\end{figure}

The target concentrates final means too tightly to respect the
receiver's initial beliefs. As Section~\ref{sec:verify-continuous} shows,
any implementation must reveal the outer-quarter states exactly and
assign the middle half to final means within $1/400$ of $1/2$.
Every receiver assigns probability $1/2$ to the middle half. Conditional
on that event, however, a receiver with initial mean $x$ expects the state to be
$\frac12+\frac14(x-\frac12)$. Since outer states are revealed, her
final means for middle states must average to this same value.
The required averages range from $0.475$ to $0.525$, whereas the target
allows only $[0.4975,0.5025]$ (Figure~\ref{fig:continuous}). They are
incompatible whenever $|x-1/2|>1/100$, which holds for nine-tenths of
receivers. This argument uses only~\eqref{eq:weak}, so the target also
remains excluded under richer initial information.

The gaps in the target's density are not essential to the failure.
By Theorem~\ref{thm:excluded-set}, mixing this target with a sufficiently
small positive weight $s$ on the uniform prior $F$ leaves it excluded.
The mixture has density at least $s$, hence positive almost everywhere.

\FloatBarrier
\section{An auction application}\label{sec:auctions}

The positive case of Theorem~\ref{thm:main} allows us to solve the
auction disclosure problem of \citet{BergemannEtAlAuctions} when bidders
initially know valuation intervals. The seller chooses how much more
each bidder learns within the interval containing her valuation.

A seller sells one object in a second-price auction without a reserve
price to $n\geq2$ risk-neutral bidders. Valuations $\theta_i$ are
independent draws from $F$ on $[a,b]$, where $0\leq a<b<\infty$ and
$F$ has a density positive almost everywhere on $(a,b)$. Each bidder
privately observes only the initial posterior mean $X_i$ generated
by the common interval partition in~\eqref{eq:partition}. Before
valuations are realized, the seller commits to a common message rule
known to all bidders. Following \citet{BergemannEtAlAuctions}, we restrict
attention to these symmetric rules with independent randomization:
$M_i$ depends only on $\theta_i$, and bidder $i$ observes only $(X_i,M_i)$.
The seller maximizes expected revenue over these rules.

Conditional on her information, each bidder's valuation is independent
of her rivals' bids. Her expected gain from winning at a given price is
therefore her final posterior mean minus that price. Bidding this mean,
$Y_i=\E[\theta_i\mid X_i,M_i]$, maximizes her expected payoff against any
bids her rivals submit using their own information.
These bids $Y_i$ are independent and identically distributed with common
law $G$. Each bidder faces rival bids drawn from this unconditional
distribution: she observes neither her rivals' initial intervals nor
their messages.

Revenue is the expected second-highest bid, so the seller evaluates
disclosure through this unconditional bid distribution.
Theorem~\ref{thm:main} makes every common bid distribution
$G\in\mathcal B(G_0,F)$ implementable, allowing the seller to choose $G$
directly.

An optimal rule reveals valuations below a cutoff within each initial
interval and pools the rest. We express the cutoffs in quantiles.
For a law $\tau$, write $\tau(t)=\tau([a,t])$ and
$\tau^{-1}(u)=\inf\{t\in[a,b]:\tau(t)\geq u\}$, $0<u<1$.
The quantile endpoints $(\alpha_j,\beta_j)=(F(\ell_j),F(r_j))$ of an
initial interval $I_j=(\ell_j,r_j)$ are the fractions of valuations
below its boundaries. A cutoff at the upper endpoint $\beta_j$ reveals
every valuation in $I_j$, and a cutoff at the lower endpoint $\alpha_j$
pools the whole interval, so that the bidder learns nothing new.

\begin{theorem}[Optimal disclosure within initial valuation intervals]
\label{thm:auction}
In the auction environment specified above, providing no further
information is optimal if $n=2$.

If $n\geq3$, put $\rho_n=(n-2)/(n-1)$. For each
$\beta\in(\rho_n,1]$, let $t_n(\beta)$ be the unique solution
$t\in(0,\rho_n)$ of
\begin{equation}\label{eq:auction-cutoff}
\int_t^\beta (u-\rho_n)u^{n-3}(u-\beta)\,du=0.
\end{equation}
For each initial interval $I_j$, set
\begin{equation}\label{eq:auction-cell-cutoff}
q_j=\begin{cases}
\beta_j,&\beta_j\leq\rho_n,\\
\max\{\alpha_j,t_n(\beta_j)\},&\beta_j>\rho_n.
\end{cases}
\end{equation}
The following deterministic message rule is optimal. Within each initial
interval $I_j$, reveal valuations $\theta_i$ with $F(\theta_i)<q_j$ and
pool the remaining valuations, disclosing only that
$F(\theta_i)\in[q_j,\beta_j]$. Send no message when $\theta_i$ lies
outside all initial intervals.
\end{theorem}

When $q_j<\beta_j$, the pooling message induces posterior mean
\[
\frac{1}{\beta_j-q_j}\int_{q_j}^{\beta_j}F^{-1}(u)\,du.
\]
Valuations outside the initial intervals are already known, so no further
message is needed there.

For $n\geq3$, the revenue effect of bid dispersion changes at the
quantile $\rho_n$ (Section~\ref{sec:auction-proof}). A mean-preserving
spread confined to quantiles below $\rho_n$ raises revenue, whereas
pooling within quantiles above it raises revenue.
At most one initial interval contains the valuation $F^{-1}(\rho_n)$.
Only this interval can require partial disclosure: the seller reveals
its lower valuations and pools its upper valuations. Lower intervals
are fully revealed; higher intervals receive no further information.
The cutoffs depend on the number of bidders and on each interval's
quantile endpoints $(\alpha_j,\beta_j)$ in the full valuation
distribution. Each interval's rule thus accounts for competition from
bidders in all initial intervals. The rules can be calculated separately
because bids from different intervals remain ordered, and the proof
expresses revenue as a sum of their contributions. The valuation
distribution then determines the corresponding valuation thresholds and
pooled means.

For example, let $n=3$, so that $\rho_3=1/2$. If bidders are initially
uninformed, the cutoff is $t_3(1)=1/4$: the seller reveals valuations in
the lowest quarter of the distribution and pools the rest. If bidders
initially know only whether their valuations lie in the top quarter,
then $t_3(3/4)=3/8$: the seller reveals valuations below the
$3/8$-quantile, pools those between the $3/8$- and $3/4$-quantiles, and
adds nothing in the top quarter.

When bidders are initially uninformed (the single interval $(a,b)$),
the rule reduces to the disclosure policy of \citet{BergemannEtAlAuctions}: reveal low valuations and pool
high valuations. The proof in Section~\ref{sec:auction-proof} adapts
the integrated-quantile argument of \citet{SorokinWinter} to the
intervals fixed by the bidders' initial information. Outside interval
partitions, some bid distributions between the bounds cannot be induced.
Optimization over the band then gives an upper bound on revenue, and
its maximizer must be checked against the conditional constraints
in~\eqref{eq:conditional-feasibility}.

\section{Discussion: Persuading the Persuaded}\label{sec:discussion}

Persuasion is easiest when the audience arrives with a clean slate: a persuader
decides what to reveal, and listeners draw their conclusions. Real audiences are
rarely so obliging. Bidders arrive at an auction with views about what the
object is worth. Voters have lived through earlier governments, and investors
have watched past promises come to nothing. Every persuader joins a conversation
already under way.

With a single listener, the outcome turns on what she ends up expecting, and the
persuader can plan separately for each possible starting point. But bidders in
an auction and tactical voters at the polls must also reckon with what everyone
else expects. What matters, then, is how expectations are spread across the
whole crowd. For an audience that starts from scratch, theory offers an elegant
shortcut. Pick any pattern of expectations that is right on average and no more
spread out than the facts themselves, and some message will produce it. The
persuader can choose the destination first and work out the message later.

It is tempting to carry the trick over to an informed audience. Two limits are
plain: further disclosure cannot leave people knowing less than they already do,
or more than the whole truth. So note where the crowd stands today and where it
would stand if it learned everything, then aim anywhere in between. The prize
would be large. Tackled head-on, the problem is daunting: the persuader must
work out what each listener can still be told, given what she already knows, and
make the pieces add up to the desired pattern. The shortcut would reduce all of
that to two snapshots of the crowd, with mathematics filling in the rest.

This paper puts the shortcut to the test. For questions whose answer can lie
anywhere on a scale, such as a price, a cost or a value, it works in exactly one
case. What people already know must sort the possibilities into clear ranges, as
when a buyer's only information about a fridge is its certified energy band
(Theorem~\ref{thm:main}). Everywhere else, some destinations that look reachable
are not. These are not isolated curiosities. They occupy whole regions, and a
small nudge to the target will not make it attainable
(Theorem~\ref{thm:excluded-set}).

Starting from scratch, a persuader can reach any destination the whole truth
allows, simply by adding information with care. With an informed audience, the
shortcut has a blind spot. It sees what people already know only through what
they expect, and an expectation is a single number. Behind it lies a whole
distribution: how likely each possible answer seems. For a yes-or-no question
nothing is lost, since the chance someone assigns to ``yes'' says it all. For a
painting's value or a fridge's running costs, the same average can conceal quite
different distributions. Filling in the rest, the mathematics may then quietly
require people to discard information they hold: a buyer who already knows that
the value is at most 100 would have to end up expecting 120.

In most real settings that is not an option: a persuader can add to what people
know, but cannot make them forget. Nor can tailoring get around the problem,
even when the persuader can see into every head and address each person
individually. Messages do not work in isolation: revealing some possibilities
changes how the rest are understood. Commit to announcing all the good news, and
silence itself becomes bad news.

Clear ranges are the exception because they divide the possible values without
overlap. Any destination allowed by the two snapshots can be split along the
same boundaries and reached by telling people more about their own range. Nobody
needs to forget anything.

This carries a lesson for the craft of influencing people with information.
Firms and campaigns increasingly predict what a person will pay or whom she will
vote for, and then try to shift those choices with well-chosen facts. But two
people with the same price in mind may hold quite different beliefs about what
they are buying. Those beliefs determine what further evidence can achieve. More
data can help uncover them, but uncovering is not erasing: a perfectly tailored
message still meets a reader with a past.

The form of knowledge therefore matters as much as its amount. Labels that
report only a range, such as a certified energy band, keep within reach every
destination consistent with what buyers now expect and with the facts. A noisy
test or a complicated score quietly closes some of them.

Mandatory disclosure is of particular interest here, because it may turn the
form of what people learn first into a policy choice. Suppose a government
requires sellers to provide information about an appliance's energy use. A rule
that sets only how informative that disclosure must be leaves sellers free to
design it themselves, and so to reach any pattern between that minimum and the
whole truth. A rule that prescribes a particular piece of information, such as a
specific test result, fixes the starting point for whatever sellers say next.
Its form then determines whether their later marketing can deliver everything
the two snapshots promise. Both rules go by the name of transparency, yet they
can make a world of difference to the information that follows.

The shortcut is thus safe in three situations. The first is an audience that
starts from scratch or whose knowledge comes in clear ranges: the persuader can
then pick the destination in advance, confident that some message will reach it
(Section~\ref{sec:auctions} illustrates this with an auction). The second is a
seller that must meet a minimum standard of informativeness but may choose how.
It can plan its disclosure and everything it says afterwards as one. The third
is a persuader who can make people forget, like Orwell's Ministry of Truth. What
they knew before then no longer matters, and every destination between the two
snapshots comes within reach, along with many beyond.

Everyone else has to take the audience as it comes and work the other way round.
They must first establish, case by case, what each listener can still be told
given what she already knows, and only then discover which destinations are
within reach. The two snapshots can still suggest an answer, but they cannot
vouch for it.

Writing on a blank slate is the easy part. Most slates are not blank, and what
is already written on them shapes what can be written next.

\section{Proofs}\label{sec:proof}

Section~\ref{sec:thm1-proof} proves Theorem~\ref{thm:main}.
The remaining subsections prove Corollary~\ref{thm:impossibility},
Theorem~\ref{thm:excluded-set}, and Corollary~\ref{cor:value},
verify the examples, and prove the auction result.

\subsection{Proof of Theorem~\ref{thm:main}}\label{sec:thm1-proof}

The proof of Theorem~\ref{thm:main} starts by identifying the states that
every target in the band must leave unchanged. The receiver already
knows these states. The remaining states form
\emph{uncertainty intervals}; each target must preserve their probabilities
and conditional means (Claim~\ref{clm:decomposition}). Under interval-partition
information, these are the disclosed intervals. Strassen's theorem then
implements any target separately within each interval
(Claim~\ref{clm:sufficiency}).

For necessity, we assume only that every target with a density in the
band admits a triple satisfying the weaker conditions~\eqref{eq:weak}.
Section~\ref{sec:necessity} proceeds in four steps:
\begin{enumerate}
\item \emph{Pool a short interval.} Pool a sufficiently short subinterval
of an uncertainty interval at its prior conditional mean. The target
lies in the band. Approximating it by targets with densities and using
closedness (Claim~\ref{clm:closedness}) gives a triple satisfying the weaker
conditions for this pooling target too (Claim~\ref{clm:pools}).
\item \emph{Identify the states behind each final mean.} For this target,
the condition $\E[\theta\mid Y]=Y$ determines the relation between
$Y$ and $\theta$: every state in the pool has the same final mean,
and all other states are revealed. Preserving each receiver's initial
mean then gives a local moment restriction on the initial experiment
(Claim~\ref{clm:local-identity}).
\item \emph{Vary the pool.} Moving its endpoints shows that the initial
posterior mean, averaged conditional on the state, is constant on each
uncertainty interval (Claim~\ref{clm:reverse-constancy}).
\item \emph{Recover the initial information.} Combining this constancy
with $\E[\theta\mid X]=X$ gives zero conditional variance of $X$
within each interval. Thus $X$ itself is constant there, giving
interval-partition information (Claim~\ref{clm:collapse}).
\end{enumerate}

\subsubsection{Call functions and uncertainty intervals}
\label{sec:geometry}\label{sec:contacts}

For a probability law $\tau$ on $[a,b]$, its \emph{call function}
records the expected amount by which a draw exceeds a threshold $t$:
\[
C_\tau(t)=\int(z-t)_+\,\tau(dz),\qquad (z-t)_+=\max\{z-t,0\}.
\]
Laws with the same mean are in convex order precisely when their call
functions are ordered. Thus $D=C_F-C_{G_0}\geq0$ measures the gap
between the two bounds. Since $D$ is continuous, $\{D>0\}$ is open;
its disjoint interval components are at most countable, as each contains
a distinct rational number. Denote them by $I_j$, $j\in J$, and call
them the \emph{uncertainty intervals}. Put
$R=[a,b]\setminus\bigcup_j I_j$. The points in $R$
are \emph{contact points}, where the two call functions agree; together
they form the \emph{contact set}. As the next claim shows,
the state is already known on $R$,
while further disclosure can change posterior means only within the
intervals $I_j$.

For a Borel set $A$, the restriction $F|_A$ is the unnormalized measure
defined by $F|_A(B)=F(A\cap B)$ for Borel sets $B$. Write
$p_j=F(I_j)$, $k_j=\int_{I_j}\theta\,F(d\theta)/p_j$, and
$F_j=F|_{I_j}/p_j$ for the probability of interval $j$, its
conditional state mean, and its conditional state distribution.
These definitions apply to every initial experiment. Under
interval-partition information, the uncertainty intervals coincide with
the disclosed intervals, as shown in Section~\ref{sec:sufficiency}.
Every triple considered below has the prescribed joint law $\pi$ of
$(X,\theta)$.

\begin{claim}[Contact points separate uncertainty intervals]\label{clm:decomposition}
For each $I_j$, the events $\{X\in I_j\}$ and $\{\theta\in I_j\}$ agree
almost surely. On $\{\theta\in R\}$, $X=\theta$ almost surely. Every law
$G$ in the two-sided band satisfies
\begin{equation}\label{eq:component}
G|_R=F|_R,\qquad G(I_j)=p_j,\qquad
\int_{I_j}y\,G(dy)=p_jk_j,\qquad
G_j:=\frac{G|_{I_j}}{p_j}\cx F_j.
\end{equation}
\end{claim}

\begin{proof}
This is the standard no-crossing decomposition
\citep[Theorem~A.4]{BeiglbockJuillet}; we give the argument needed here.
We first show that equality of call functions at a threshold prevents a
martingale coupling from moving mass across it. Applying this observation
at the contact points will separate the uncertainty intervals.
If integrable variables $A,B$ satisfy $\E[B\mid A]=A$ and have equal call
functions at $t$, conditional Jensen is an equality for $(B-t)_+$.
It follows that $A<t$ implies $B\leq t$, $A>t$ implies $B\geq t$, and
$A=t$ implies $B=t$. The last assertion also uses $\E[B\mid A]=A$.
If the law of $B$ is atomless, both variables have zero mass at $t$.

Apply this to $(X,\theta)$ at every endpoint of the components $I_j$ of
$\{D>0\}$. It gives $\{X\in I_j\}=\{\theta\in I_j\}$ almost surely, so
these events have mass $p_j$ and conditional mean $k_j$.
On $R$, we have $X=\theta$ almost surely. To show this simultaneously on a
possibly uncountable contact set, let $R_0\subseteq R$ be a countable dense
subset that contains all component endpoints. Using a countable set lets
us impose all these no-crossing restrictions outside a single null event.
Almost surely, no point of $R_0$ lies strictly between $X$ and $\theta$, and $\theta$ is not a
component endpoint. Suppose $\theta\in R$ and $X\neq\theta$. If the open
interval between them contained no point of $R$, it would lie in a single
component, and $\theta$ would be an endpoint of that component. Hence the
interval contains a point of $R$, and therefore a point of $R_0$,
contradicting the no-crossing property.
This proves diagonality on $R$, and hence $G_0|_R=F|_R$.

For any $G$ in the band, $C_G=C_F$ on $R$. Choose a
martingale coupling from $G$ to $F$ and apply the same cuts.
This gives the restriction, mass, and first-moment identities
in~\eqref{eq:component}. The convex-order inequality follows from
conditional Jensen within each component.
\end{proof}

\paragraph{The one-sided problems within intervals.}
Under interval-partition information, every implementable law, and every law in the
band, has the form
\begin{equation}\label{eq:feasible}
\boxed{\quad G=F|_R+\sum_{j\in J}p_jG_j,
\qquad G_j\cx F_j\quad(j\in J).\quad}
\end{equation}
Here each $G_j$ is a freely chosen probability law of posterior
means within interval $j$. Convex order fixes its mean at $k_j$ and
restricts its support to that interval. There is no separate lower restriction to check:
the initial law there is just $\delta_{k_j}$, unit probability at $k_j$,
and every law with that mean is above it in convex order.
The sender can therefore choose each $G_j$ by solving a one-sided
problem with conditional state distribution $F_j$, then combine those
choices with weights $p_j$. Outside interval-partition information,
receivers whose states lie in the same uncertainty interval may have
different initial posterior means. Matching that interval's probability
and conditional mean alone does not ensure that each receiver's initial
information remains consistent with the target.

\subsubsection{Interval-partition information: every target is implementable}
\label{sec:sufficiency}

\begin{claim}[Implementation within the disclosed intervals]\label{clm:sufficiency}
If $\pi$ has the interval-partition form~\eqref{eq:partition}, every law
$G_0\cx G\cx F$ is implementable. The message rule can depend
on $\theta$ alone.
\end{claim}

\begin{proof}
We first identify the disclosed intervals with the uncertainty intervals,
so that Claim~\ref{clm:decomposition} describes each target's conditional
parts. We can then implement those parts separately.
Suppose \eqref{eq:partition} holds with intervals $I_j=(\ell_j,r_j)$,
and write $p_j=F(I_j)$ and $k_j=\E_\pi[\theta\mid \theta\in I_j]$. Then $X$ and
$\theta$ lie in the same \emph{cell} of the partition: the same disclosed
interval, or the same single state outside all such intervals. Thus $\{X=k_j\}=\{\theta\in I_j\}$
almost surely. For $t\notin\bigcup_jI_j$, every cell lies on one side of
$t$, so $\E[\theta\mid X]=X$ gives $D(t)=0$. For $t\in I_j$, the other cells
contribute nothing, and
\[
D(t)=\int_{I_j}(\theta-t)_+\,F(d\theta)-p_j(k_j-t)_+
=\begin{cases}
\int_{I_j}(\theta-t)_+\,F(d\theta),&t\geq k_j,\\[2pt]
\int_{I_j}(t-\theta)_+\,F(d\theta),&t<k_j,
\end{cases}
\]
which is positive because $f>0$ almost everywhere on $I_j$. The
intervals $I_j$ are therefore the uncertainty intervals.

Let $G\in\mathcal B(G_0,F)$, written as in~\eqref{eq:feasible}.
For each $j$, Strassen's theorem gives a pair $(Y_j,\theta_j)$ with
$Y_j\sim G_j$, $\theta_j\sim F_j$, and $\E[\theta_j\mid Y_j]=Y_j$. If $\theta\notin\bigcup_jI_j$, send $M=\theta$. If $\theta\in I_j$, draw
$M$ from the conditional law of $Y_j$ given $\theta_j=\theta$.
It remains to check that retaining $X$ does not change the posterior mean
induced by $M$. A martingale coupling can place $Y_j$ at an endpoint of $I_j$ only if $\theta_j$ equals that endpoint,
which has probability zero; hence $M$ lies in the same cell as $\theta$ and
reveals $X$. Consequently $\E[\theta\mid X,M]=\E[\theta\mid M]=M$, so the final
posterior mean is $M$ and its law is $F|_R+\sum_jp_jG_j=G$.
The message rule uses only $\theta$.
\end{proof}

\subsubsection{Other initial experiments: an excluded target with a density}
\label{sec:necessity}

Consider triples $(X,Y,\theta)$ that satisfy~\eqref{eq:weak} and have
$(X,\theta)\sim\pi$. Let $\mathcal A(\pi)$ be the set of their marginal
laws $G=\Law(Y)$. Recall from Section~\ref{sec:richer} that
these conditions hold whenever $Y$ is the final posterior mean of a
receiver whose initial information has posterior mean $X$, and that
they follow from~\eqref{eq:implementation}. Hence
\[
\mathcal I(\pi)\subseteq\mathcal A(\pi)\subseteq\mathcal B(G_0,F),
\]
where the second inclusion follows from conditional Jensen. Throughout
this part of the proof, assume that every absolutely continuous target in the
band belongs to $\mathcal A(\pi)$; we show that $\pi$ then generates
interval-partition information. The assumption is weaker than
implementability of every target in two ways: it concerns only targets
with densities, and it does not require the final posterior mean to be
the expected state given both $X$ and $Y$.

The argument uses limits of targets in the weak topology. The next claim
shows that $\mathcal A(\pi)$ contains the limits of its elements; the
same property of $\mathcal I(\pi)$ is used in the proof of
Theorem~\ref{thm:excluded-set}.

\begin{claim}[Closedness of $\mathcal A(\pi)$ and $\mathcal I(\pi)$]\label{clm:closedness}
The sets $\mathcal A(\pi)$ and $\mathcal I(\pi)$ are convex and weakly
compact, and therefore weakly closed.
\end{claim}

\begin{proof}
To pass to weak limits, we write the conditional-mean restrictions as
integral identities with continuous integrands.
Probability laws $Q$ of
$(x,y,\theta)$ on the compact space $[a,b]^3$ form a weakly compact space;
fixing their $(x,\theta)$ marginal to $\pi$ is a closed condition.
Weak convergence preserves the following identities because their
integrands are continuous and bounded:
\[
\int(y-x)h(x)\,dQ=0,\qquad
\int(\theta-y)h(y)\,dQ=0
\quad\text{for every continuous }h:[a,b]\to\mathbb R.
\]
These integral conditions are equivalent to \eqref{eq:weak},
because continuous functions determine finite signed measures.
The continuous projection to the $y$ marginal therefore gives the compact set
$\mathcal A(\pi)$. Replacing the second condition by
$\int(\theta-y)g(x,y)\,dQ=0$ for every continuous $g:[a,b]^2\to\mathbb R$
imposes \eqref{eq:implementation}; the same argument gives compactness
of $\mathcal I(\pi)$. Both sets are convex because the restrictions
on $Q$ are linear and taking a marginal is a linear operation.
\end{proof}

\paragraph{Pooling targets.}
The targets used below replace the states in a short interval by their
conditional mean and leave all other states unchanged. Taking the
interval sufficiently short keeps the resulting law above the lower
bound, while pooling automatically respects the upper bound.
Fix an uncertainty interval $I$ and a nondegenerate compact interval
$\widehat I\subset I$. Since $D$ is continuous and positive on
$\widehat I$, compactness gives a positive minimum
$d_{\widehat I}=\min_{t\in\widehat I}D(t)$. This uniform gap leaves room
to lower the prior's call function by pooling without crossing the lower
bound. Call $K=(u,v)$ an \emph{admissible pool} if $[u,v]\subset\widehat I$ and
$0<v-u<d_{\widehat I}$, and let
\[
m_K=\frac{\int_K \theta\,F(d\theta)}{F(K)},\qquad
T_K(\theta)=\begin{cases}m_K,&\theta\in K,\\\theta,&\theta\notin K,\end{cases}
\qquad G_K=\Law(T_K(\theta)).
\]

\begin{claim}[Small pooling targets belong to $\mathcal A(\pi)$]\label{clm:pools}
For every admissible pool $K$, the law $G_K$ lies in the two-sided
band and belongs to $\mathcal A(\pi)$.
\end{claim}

\begin{proof}
Pooling preserves the mean and reduces every convex expectation, so
$G_K\cx F$. To check the lower bound, compare the loss in the call
function with the available gap $d_{\widehat I}$. The loss is zero at
thresholds outside $[u,v]$, where the call payoff is affine on the pool.
Inside, it is at most $(v-u)F(K)<d_{\widehat I}$. Thus pooling cannot
lower the call function below $C_{G_0}$, giving $G_0\cx G_K$.

Since $G_K$ has an atom, the standing assumption, which concerns targets
with densities, does not apply to it directly. We therefore retain a
small amount of variation within the pool and then let it vanish.
Replace $T_K(\theta)$ inside $K$ by
$m_K+\varepsilon(\theta-m_K)$, where $0<\varepsilon<1$, leaving it unchanged
outside. Denote the resulting law by $G_{K,\varepsilon}$.
It has a density, since both pieces of the map are affine with nonzero
slope, and
\[
G_0\cx G_K\cx G_{K,\varepsilon}\cx F,
\qquad G_{K,\varepsilon}\Rightarrow G_K
\quad\text{as }\varepsilon\downarrow0.
\]
The convex orders follow from Jensen and convexity for contraction toward
$m_K$; the arrow denotes weak convergence. Each
$G_{K,\varepsilon}$ belongs to $\mathcal A(\pi)$ by the assumption
of this part of the proof. Claim~\ref{clm:closedness} therefore gives
$G_K\in\mathcal A(\pi)$.
\end{proof}

\paragraph{The reverse conditional mean.}
The restriction that a pooling target imposes on the initial experiment
is expressed through the function
\begin{equation}\label{eq:reverse-mean}
w_\pi(\theta)=\E_\pi[X\mid \theta],
\end{equation}
which averages the receiver's initial posterior mean over the observations
that can occur when the state is $\theta$. It is determined by the initial
experiment. We choose a measurable representative with values in $[a,b]$
everywhere: conditional expectation determines it only $F$-almost
everywhere, and values on the exceptional null set do not affect the
integrals below. Conditioning on $\theta$ will let us use $w_\pi$ to
express the moment restriction on $(X,Y,\theta)$ as an integral over states.
Under interval-partition information, the initial posterior mean equals
$\E_\pi[\theta\mid \theta\in I_j]$ whenever the state lies in $I_j$, so $w_\pi$ is
constant on each disclosed interval. In the continuous example of
Section~\ref{sec:continuous-example}, it is increasing:
$w_\pi(\theta)=\frac12+\frac1{25}(\theta-\frac12)$. The remaining claims show
that, if every target with a density in the band belongs to $\mathcal A(\pi)$,
then $w_\pi$ is constant on each uncertainty interval, and $X$ itself is
constant there as well.

Specifying the distribution $G_K$ alone does not specify which states
give rise to which final means. We must also use $\E[\theta\mid Y]=Y$.
The next claim shows that, among joint laws with marginals $G_K$ and $F$,
this condition allows only $Y=T_K(\theta)$ almost surely. Every
implementation must therefore give final mean $m_K$ to the states in
$K$ and reveal the rest. This determines the final means, although
different message rules may produce them. The resulting restriction on
$w_\pi$ applies to every possible implementation.

\begin{claim}[A necessary identity for pooling targets]\label{clm:local-identity}
Any coupling of $Y\sim G_K$ and $\theta\sim F$ with $\E[\theta\mid Y]=Y$
must satisfy $Y=T_K(\theta)$ almost surely. Consequently, for every admissible
pool $K=(u,v)$,
\begin{equation}\label{eq:local}
\int_u^v(\theta-m_K)w_\pi(\theta)f(\theta)\,d\theta=0.
\end{equation}
\end{claim}

\begin{proof}
We need a convex function with two properties: pooling inside $K$
leaves its expectation unchanged, and equality in Jensen when
$Y\notin[u,v]$ is possible only if $\theta=Y$.
Choose one that is flat on $[u,v]$ and strictly convex on either side:
the continuously differentiable function
\[
\phi(\theta)=(u-\theta)_+^2+(\theta-v)_+^2.
\]
Since $\phi$ vanishes on $[u,v]$, where the two laws differ, its
expectations under them agree. Because $\E[\theta\mid Y]=Y$, the
nonnegative variable $\phi(\theta)-\phi(Y)-\phi'(Y)(\theta-Y)$ therefore
has expectation zero and vanishes almost surely. Thus $\theta$ must
be a point where the tangent at $Y$ touches the graph of $\phi$.
If $Y$ lies outside $[u,v]$, the only such point is $Y$, so $\theta=Y$.
If $Y$ lies inside $[u,v]$, its tangent is the zero line, which touches
the graph exactly on $[u,v]$; hence $\theta\in[u,v]$. There is no endpoint mass,
and the only target value inside is $m_K$. Hence precisely the states
in $K$ must receive $m_K$; everything else is revealed.

Now impose the prescribed initial law. By Claim~\ref{clm:pools}, choose
a triple with $(X,\theta)\sim\pi$, $\Law(Y)=G_K$, and~\eqref{eq:weak}.
Both the state and the final mean have conditional mean $X$ given $X$,
so $\E[X\theta]=\E[X^2]=\E[XY]$. Since $Y=T_K(\theta)$ almost surely,
$\theta-Y$ is a function of the state. Conditioning on $\theta$ gives
\[
0=\E[X(\theta-Y)]
=\int_K(\theta-m_K)w_\pi(\theta)\,F(d\theta),
\]
which is~\eqref{eq:local}.
\end{proof}

Condition~\eqref{eq:local} says that the state and the reverse conditional
mean are uncorrelated within every admissible pool. To see what this
rules out, suppose $w_\pi$ were strictly increasing there. Pooling lowers final
means for states above $m_K$ and raises them for states below $m_K$.
With $w_\pi$ increasing, states above $m_K$ come with higher initial
means on average, so pooling would reduce $\E[XY]$ below $\E[X\theta]$,
violating the martingale identities.
The next claim turns this intuition into a constancy result without
assuming that $w_\pi$ is monotone or continuous.

\begin{claim}[The reverse conditional mean is constant on each uncertainty interval]
\label{clm:reverse-constancy}
For each uncertainty interval $I$, there is a constant $c_I$ such that
$w_\pi(\theta)=c_I$ for Lebesgue-almost every $\theta\in I$.
\end{claim}

\begin{proof}
We will show that the values of $w_\pi$ at the two endpoints of a short
pool both equal its average within that pool. They must therefore equal
each other. Joining nearby endpoints will then give constancy throughout
an uncertainty interval.

Fix a nondegenerate compact interval $\widehat I\subset I$. Every pool
inside it of length less than $d_{\widehat I}$ is admissible, so the
identity from Claim~\ref{clm:local-identity} holds as we move its endpoints
slightly. To differentiate this identity, write its terms as
\[
\begin{aligned}
A_0&=\int_u^v f(\theta)\,d\theta,&A_1&=\int_u^v \theta f(\theta)\,d\theta,\\
B_0&=\int_u^v w_\pi(\theta)f(\theta)\,d\theta,&B_1&=\int_u^v \theta w_\pi(\theta)f(\theta)\,d\theta.
\end{aligned}
\]
These quantities depend on $u,v$. Here $A_0$ is the pool's probability,
$m_K=A_1/A_0$ its mean state, and $\bar w=B_0/A_0$ its average of
$w_\pi$. Claim~\ref{clm:local-identity} becomes
\[
B_1=\frac{A_1B_0}{A_0}.
\]

We differentiate the integrals, not $w_\pi$ itself, so continuity of
$w_\pi$ is unnecessary. The four integrands are integrable; their
indefinite integrals are therefore absolutely continuous. Choose a
common full-measure set $E\subset I$ of Lebesgue points of all four
integrands where $f>0$. At endpoints in $E$, moving an upper integration
limit has derivative equal to the integrand there; moving a lower limit
has the opposite derivative. Using this same set for both endpoints
will let us compare and join the resulting equalities.

Take $u<v$ in $E$ and in the interior of $\widehat I$, with
$v-u<d_{\widehat I}$. Moving the upper endpoint exposes $w_\pi(v)$.
The product and quotient rules give
\[
\frac{\partial B_1}{\partial v}=v w_\pi(v)f(v),
\qquad
\frac{\partial}{\partial v}\left(\frac{A_1B_0}{A_0}\right)
=f(v)\bigl[v\bar w+m_Kw_\pi(v)-m_K\bar w\bigr].
\]
Equating these derivatives and collecting terms yields
\[
(v-m_K)f(v)\bigl(w_\pi(v)-\bar w\bigr)=0.
\]
Because $f>0$ almost everywhere on the pool, its mean lies strictly
between its endpoints: $u<m_K<v$. Also $f(v)>0$ by our choice of $E$.
Thus $w_\pi(v)=\bar w$. Moving the lower endpoint gives the other
endpoint value in exactly the same way:
\[
(m_K-u)f(u)\bigl(w_\pi(u)-\bar w\bigr)=0,
\qquad\text{hence}\qquad w_\pi(u)=\bar w=w_\pi(v).
\]

This equality initially compares only nearby points. Since $E$ has full
measure, it is dense. Any two of its points in the interior of
$\widehat I$ can be joined by a finite increasing sequence of points in
$E$ whose successive gaps are less than $d_{\widehat I}$. Applying the
equality to each consecutive pair shows that $w_\pi$ is constant on
$E\cap\operatorname{int}(\widehat I)$, hence almost everywhere on
$\widehat I$. Finally, take an increasing sequence of such compact
intervals covering $I$. Their constants agree on their overlaps, which
have positive length. This gives a single constant $c_I$ almost
everywhere on $I$.
\end{proof}

\begin{claim}[Each uncertainty interval has one initial mean]\label{clm:collapse}
For every uncertainty interval $I$, let $p=F(I)$ and
$k=\int_I \theta\,F(d\theta)/p$. Then $X=k$ almost surely on $\{\theta\in I\}$.
\end{claim}

\begin{proof}
Constancy of $w_\pi$ says that initial means have the same average at
almost every state in $I$. To show that $X$ itself is constant there, we use
the martingale property to show that its conditional variance is zero.
By Claim~\ref{clm:decomposition}, the event
$B=\{X\in I\}$ equals $\{\theta\in I\}$ almost surely. We can therefore
condition on either $X$ or $\theta$ while keeping the event $B$ fixed.
With $\1_B$ denoting its indicator, the initial martingale and
Claim~\ref{clm:reverse-constancy} give
\[
\E[X\1_B]=\E[\theta\1_B]=pk,
\qquad \E[X\1_B]=pc_I,
\]
so $c_I=k$. The same identities determine the second moment:
\[
\E[X^2\1_B]=\E[X\theta\1_B]
=\E[w_\pi(\theta)\theta\1_B]=pk^2.
\]
Consequently $\E[(X-k)^2\1_B]=0$, which proves the claim.
\end{proof}

\begin{proof}[Completion of the proof of Theorem~\ref{thm:main}]
Claim~\ref{clm:collapse} holds on every uncertainty interval, while
Claim~\ref{clm:decomposition} gives $X=\theta$ on $R$. Since there are at most
countably many uncertainty intervals, these conclusions hold simultaneously outside
a single null event. Thus $\pi$ has the interval-partition
form~\eqref{eq:partition}, with the uncertainty intervals as its
intervals.

We have shown that if every absolutely continuous target in the band
belongs to $\mathcal A(\pi)$, then $\pi$ generates interval-partition
information. Contrapositively, for every other initial experiment some
absolutely continuous target in the band lies outside $\mathcal A(\pi)$,
and hence outside $\mathcal I(\pi)$. Together with
Claim~\ref{clm:sufficiency}, this proves Theorem~\ref{thm:main}.
\end{proof}

The proof also shows where the density assumption is used. The necessity
argument uses positivity of the density only inside the uncertainty
intervals, so it applies whenever $F$ is absolutely continuous and its
density is positive almost everywhere on each uncertainty interval.
Absolute continuity of $F$ remains essential for the conclusion about
targets with densities: by Claim~\ref{clm:decomposition}, every target in
the band agrees with $F$ on the contact set $R$, so a singular part of
$F$ on $R$ would appear in every target.

\subsection{Proof of Corollary~\ref{thm:impossibility}}\label{sec:atomless-proof}

\begin{proof}
Under interval-partition information, each disclosed interval
would put positive mass at its conditional mean, contradicting
atomlessness of $G_0$. With no such interval, $X=\theta$ and $G_0=F$.
Thus atomlessness of $G_0$ together with $G_0\ne F$ rules out
interval-partition information. The necessity argument in
Section~\ref{sec:necessity} therefore supplies an absolutely continuous
target $G\in\mathcal B(G_0,F)\setminus\mathcal A(\pi)$, which is not
implementable because $\mathcal I(\pi)\subseteq\mathcal A(\pi)$.

If the receiver observes a richer signal $S$ with $\E[\theta\mid S]=X$
and the sender observes $S$ as well, any implementation would still
produce a triple satisfying~\eqref{eq:weak} (Section~\ref{sec:richer}).
It would therefore place $G$ in $\mathcal A(\pi)$, a contradiction.
The same target remains excluded under richer initial information.
\end{proof}

\subsection{Proof of Theorem~\ref{thm:excluded-set}}\label{sec:thm2-proof}

The implementable set is closed, so implementable targets cannot converge
to an excluded target. Fix an excluded target and mix in another target from the band. As the
weight on the second target tends to zero, the mixture converges to the
excluded target, uniformly over the choice of the second target.

\begin{proof}
Fix $G_*\in\mathcal B(G_0,F)\setminus\mathcal I(\pi)$; this set is
nonempty by Theorem~\ref{thm:main}, and $\mathcal I(\pi)$ is weakly
closed by Claim~\ref{clm:closedness}.
Mixing two targets in the band keeps the mixture in the band by convexity.
If no uniform $\bar s>0$ existed, there would be $s_n\to0$ and
$G_n\in\mathcal B(G_0,F)$ for which
$(1-s_n)G_*+s_nG_n$ is implementable. For every continuous
$h:[a,b]\to\mathbb R$,
\[
\left|\int h\,d\bigl[(1-s_n)G_*+s_nG_n-G_*\bigr]\right|
\le 2s_n\|h\|_\infty\longrightarrow0.
\]
The mixtures converge weakly to $G_*$, contradicting closedness.

The set $\mathcal B(G_0,F)\setminus\mathcal I(\pi)$ is nonempty by
Theorem~\ref{thm:main} and relatively open in $\mathcal B(G_0,F)$
because $\mathcal I(\pi)$ is closed. Take any Borel probability distribution with support
$\mathcal B(G_0,F)$. Such a distribution assigns positive probability
to every nonempty relatively open subset of the band, including the
excluded set.
\end{proof}

\subsection{Proof of Corollary~\ref{cor:value}}

\begin{proof}
Under interval-partition information, Theorem~\ref{thm:main} gives
$\mathcal I(\pi)=\mathcal B(G_0,F)$, so the values agree for every $U$.
Otherwise, choose an excluded target $G_*$. By
Claim~\ref{clm:closedness}, $\mathcal I(\pi)$ is convex and weakly compact.
Strict separation in the weak topology therefore gives a continuous
$U:[a,b]\to\mathbb R$ such that
\[
\max_{G\in\mathcal I(\pi)}\int U\,dG
<\int U\,dG_*
\leq\max_{G\in\mathcal B(G_0,F)}\int U\,dG.
\]
\end{proof}

\subsection{Verification of the discrete example}\label{sec:verify-discrete}

Use the notation of Section~\ref{sec:discrete-example}, where both
bounds are verified. The following calculation shows why the two
constructions cannot be combined while keeping the prescribed initial
joint law.

Suppose a triple with $(X,\theta)\sim\pi$ and $\Law(Y)=G$
satisfied~\eqref{eq:weak}. Since
$C_G(1/2)=C_F(1/2)=1/8$, equality in Jensen prevents any
crossing of $1/2$. Since the target assigns probability only to $1/4$
and $3/4$, this no-crossing property gives $Y=1/4$ when $\theta<1/2$
and $Y=3/4$ when $\theta>1/2$. But $\Pr(\theta<1/2\mid L)=9/10$, so
\[
\E[Y\mid L]=\tfrac9{10}\tfrac14+\tfrac1{10}\tfrac34
=\frac3{10}\ne\frac{13}{50}.
\]
This contradicts $\E[Y\mid X]=X$ because $X$ identifies $L$.

To obtain an explicit bound on the mixing weight, measure the distance from
the two target values by $\psi(y)=\min\{|y-1/4|,|y-3/4|\}$.
Any implementable $Y$ lies in $[0,3/5]$ after $L$ and in
$[2/5,1]$ after $H$. On these respective supports,
\[
\psi(y)\ge\tfrac37(y-1/4),\qquad
\psi(y)\ge\tfrac37(3/4-y).
\]
The conditional means therefore give $\E[\psi(Y)]\ge3/700$:
every implementable target must stay a positive average distance from
the two proposed values.
Since $\int\psi\,dG=0$ and $\psi\le1/4$ on $[0,1]$, the expectation
of $\psi$ under $(1-s)G+s\eta$ is at most $s/4<3/700$ for
$0<s<3/175$, whatever $\eta\in\mathcal B(G_0,F)$.

\subsection{Verification of the continuous example}\label{sec:verify-continuous}

Integration of \eqref{eq:joint} gives the stated uniform marginals,
$\E[\theta\mid X]=X$, and
\begin{equation}\label{eq:reverse}
\E[X\mid \theta]=\frac12+\frac1{25}(\theta-1/2).
\end{equation}
The conditional density of $\theta$ given $X=x$ is
$1+12(x-1/2)(\theta-1/2)$. For $x_1<x_2$ and $\theta_1<\theta_2$, the product
of the densities at $(x_2,\theta_2)$ and $(x_1,\theta_1)$ exceeds the product
at $(x_2,\theta_1)$ and $(x_1,\theta_2)$ by
$12(x_2-x_1)(\theta_2-\theta_1)>0$. Since the densities are positive, the
density ratio for the higher initial mean is strictly increasing in
the state, as claimed.

\paragraph{The two bounds.}
To obtain the law in \eqref{eq:target-cont}, leave $\theta$ outside
$[1/4,3/4]$ unchanged and map the middle interval to
$1/2+(\theta-1/2)/100$. Contraction toward the conditional mean reduces
every convex expectation, so $G\cx F$.

For the lower bound, let
$\widetilde G=\tfrac12\delta_{2/5}+\tfrac12\delta_{3/5}$.
Convexity gives $G_0\cx\widetilde G$. For $2/5\leq t\leq3/5$, the upper
tail of $G$ alone gives
\[
C_G(t)-C_{\widetilde G}(t)
\geq\int_{3/4}^1(y-t)\,dy-\frac12\left(\frac35-t\right)
=\frac t4-\frac{13}{160}\geq\frac3{160}>0.
\]
For $t\leq2/5$, $C_{\widetilde G}(t)=1/2-t\leq C_G(t)$ by the
common mean; for $t\geq3/5$, $C_{\widetilde G}(t)=0$.
Thus $G_0\cx\widetilde G\cx G\cx F$.

\paragraph{Failure of the weaker conditions.}
Suppose a triple with initial joint law \eqref{eq:joint} and target
\eqref{eq:target-cont} satisfies \eqref{eq:weak}. The convex function
\[
\phi(\theta)=(1/4-\theta)_+^2+(\theta-3/4)_+^2
\]
has equal expectations under $F$ and $G$. The nonnegative
Jensen gap $\phi(\theta)-\phi(Y)-\phi'(Y)(\theta-Y)$ therefore vanishes almost
surely. Thus $\theta$ lies where the tangent to $\phi$ at $Y$ touches
the graph. If $Y$ lies outside $[1/4,3/4]$, the only such point is $Y$,
so $\theta=Y$. If $Y$ lies inside that interval, the tangent is the zero
line, which touches the graph exactly on $[1/4,3/4]$; hence $\theta$
lies there as well. Consequently,
\[
\theta=Y\text{ outside }[1/4,3/4],\qquad
E:=\{\theta\in[1/4,3/4]\}=\{Y\in[1/4,3/4]\}
\quad\text{a.s.}
\]
There is no boundary mass. On $E$, which has probability $1/2$, the
target requires $|Y-1/2|\leq1/400$.

The two identities $\E[\theta\mid X]=\E[Y\mid X]=X$ imply
$\E[(X-1/2)(\theta-Y)]=0$. Since the contribution outside $E$ is zero,
\begin{equation}\label{eq:required-moment}
\E[(X-1/2)(\theta-1/2)\1_E]
=\E[(X-1/2)(Y-1/2)\1_E].
\end{equation}
By \eqref{eq:reverse}, the left side equals
\[
\frac1{25}\int_{1/4}^{3/4}(\theta-1/2)^2\,d\theta=\frac1{2400}.
\]
The absolute value of the right side is at most
\[
\frac1{10}\cdot\frac1{400}\cdot\frac12=\frac1{8000}<\frac1{2400},
\]
contradicting \eqref{eq:required-moment}.

\subsection{Proof of Theorem~\ref{thm:auction}}\label{sec:auction-proof}

\paragraph{The bounds in quantile form.}
The auction proof expresses revenue in terms of bid quantiles. We first
translate the two bounds into this representation.
Using the cumulative distribution and quantile notation introduced in
Section~\ref{sec:auctions}, define the \emph{integrated quantile}
\[
V_\tau(q)=\int_0^q\tau^{-1}(u)\,du,\qquad 0\leq q\leq1.
\]
The band is $V_F\leq V_G\leq V_{G_0}$. At a contact threshold
$t$, all three cumulative probabilities equal $q_t=F(t)$, and
\[
V_F(q_t)=V_G(q_t)=V_{G_0}(q_t).
\]
Here $t$ is a state threshold and $q_t$ is the probability below it.
An uncertainty interval $I_j=(\ell_j,r_j)$ corresponds to a probability
segment of length $p_j$, across which each integrated quantile increases
by $p_jk_j$. The two bounding curves agree at the endpoints of that
segment. Changing the target within the band can change the curve
between those endpoints, but not the interval's probability or its
conditional mean.

\begin{proof}
The argument of \citet[Appendix, proof of Theorem~1]{SorokinWinter}
expresses revenue as a linear function of the integrated quantile.
The sign of its coefficient determines where full disclosure or pooling
raises revenue. We apply this argument separately within each initial
valuation interval, then optimize the cutoff in the interval where the
coefficient changes sign.

Let $G$ be the common law of final posterior means. Independence
and risk neutrality make bidding that mean optimal. Represent the bids
as $G^{-1}(U_i)$ for independent uniforms $U_i$. The second-highest
$U_i$ has density $n(n-1)u^{n-2}(1-u)$, giving
\begin{equation}\label{eq:auction-revenue}
\operatorname{Rev}_n(G)=n(n-1)\int_0^1u^{n-2}(1-u)G^{-1}(u)\,du.
\end{equation}
This representation includes bid distributions with atoms. For $n\geq3$,
integration by parts yields
\begin{equation}\label{eq:auction-integrated-revenue}
\operatorname{Rev}_n(G)=n(n-1)^2\int_0^1
(u-\rho_n)u^{n-3}V_G(u)\,du.
\end{equation}
The boundary terms vanish. The coefficient of $V_G$ is negative
below $\rho_n$ and positive above it.

\medskip\noindent\emph{Disclosure within an initial interval.}
By~\eqref{eq:feasible}, the conditional bid distributions can be chosen
independently across the initial intervals. Their supports remain in
those intervals, so their positions in the common bid quantile are fixed.
For one such quantile interval $[\alpha,\beta]$, write $V=V_G$. The initial integrated quantile on this interval is
the chord joining $(\alpha,V_F(\alpha))$ to $(\beta,V_F(\beta))$.
Feasibility therefore requires $V$ to be convex, to agree with $V_F$ at
both endpoints, and to lie between $V_F$ and this chord. Conversely,
these conditions describe the conditional convex-order constraint.
The density assumption makes $F^{-1}$ continuous and strictly
increasing, so $V_F$ is continuously differentiable and strictly convex.

If $\beta\leq\rho_n$, the coefficient in
\eqref{eq:auction-integrated-revenue} is nonpositive throughout the
interval, and choosing $V=V_F$ maximizes its contribution. This is full
disclosure. If $\alpha\geq\rho_n$, choosing the upper chord maximizes
the contribution. This leaves the initial interval pooled.

Suppose $\alpha<\rho_n<\beta$. For any feasible $V$, let $L$ be the
line through $(\rho_n,V(\rho_n))$ and $(\beta,V_F(\beta))$.
Convexity gives $L\leq V$ to the left of $\rho_n$ and $L\geq V$
between $\rho_n$ and $\beta$. Replacing $V$ by $\max\{V_F,L\}$ therefore lowers the integrated quantile
where the revenue coefficient is negative and raises the integrated
quantile where the coefficient is positive. The replacement is convex and
stays between $V_F$ and the upper chord, with the same endpoint values.
It is thus feasible and weakly increases revenue.

Strict convexity of $V_F$ implies that $L$ first meets $V_F$ at some
$c\in[\alpha,\rho_n]$ and remains above it until $\beta$.
Consequently the replacement has the form
\begin{equation}\label{eq:auction-chord}
V_c(u)=\begin{cases}
V_F(u),&\alpha\leq u\leq c,\\
V_F(c)+\dfrac{V_F(\beta)-V_F(c)}{\beta-c}(u-c),&c<u\leq\beta.
\end{cases}
\end{equation}
The derivative of $V_c$ is the valuation quantile below $c$ and a
constant above $c$. To implement the bid law described by $V_c$, reveal
valuations with quantiles below $c$ and pool the remaining valuations in
this initial interval. The geometric
replacement establishes that some rule of this form is optimal among
all feasible rules, including randomized ones.

\medskip\noindent\emph{Choosing the cutoff.}
Let $\mathcal R(c)$ denote the contribution of $[\alpha,\beta]$ to
\eqref{eq:auction-integrated-revenue} under~\eqref{eq:auction-chord}.
Differentiating gives
\begin{equation}\label{eq:auction-cutoff-derivative}
\mathcal R'(c)=n(n-1)^2
\frac{V_F(\beta)-V_F(c)-(\beta-c)V_F'(c)}{(\beta-c)^2}
\kappa_n(c,\beta),
\end{equation}
where
\[
\kappa_n(c,\beta)=\int_c^\beta(u-\rho_n)u^{n-3}(u-\beta)\,du.
\]
The moving-boundary terms cancel because the two pieces of $V_c$
agree at $c$. Differentiation is valid on $c<\beta$ because $V_F$ is
continuously differentiable and the integrands and their derivatives
are bounded on compact subintervals. Strict convexity makes the
fraction in~\eqref{eq:auction-cutoff-derivative} positive.

For every $\beta>\rho_n$,
\[
\kappa_n(0,\beta)=
\frac{\beta^{n-1}[n-(n-1)\beta]}{n(n-1)^2}>0,
\qquad \kappa_n(\rho_n,\beta)<0,
\]
and
\[
\frac{\partial \kappa_n}{\partial c}
=-(c-\rho_n)c^{n-3}(c-\beta)<0
\qquad(0<c<\rho_n).
\]
Thus~\eqref{eq:auction-cutoff} has exactly one root $t_n(\beta)$
in $(0,\rho_n)$. Revenue increases up to this root and decreases
thereafter on $[\alpha,\rho_n]$. If the root lies below $\alpha$,
the constrained optimum is $\alpha$. This proves
\eqref{eq:auction-cell-cutoff} for an interval crossing $\rho_n$;
for an interval above $\rho_n$, the same formula gives $q_j=\alpha_j$.

Only one of the disjoint quantile intervals $(\alpha_j,\beta_j)$ can
contain $\rho_n$. Combining
the policies therefore gives the stated rule, with no change on
singleton cells. Each conditional bid law is below its conditional
valuation law in convex order, so the combined rule is feasible
by~\eqref{eq:feasible}.
The argument also covers countably many intervals: valuations are
bounded, the revenue integrals converge absolutely, and the rule is
measurable because the collection of nonsingleton cells is countable.

Finally, for $n=2$, integration by parts in~\eqref{eq:auction-revenue}
gives $\operatorname{Rev}_2(G)=2\int_0^1V_G(u)\,du$. Since
$V_G\leq V_{G_0}$, the initial bid distribution $G_0$ maximizes
revenue and is implemented by providing no further information.
\end{proof}

\clearpage
\begingroup\small\raggedright
\setlength{\bibsep}{3pt plus 1pt}

\endgroup
\end{document}